\documentclass[authoryear,a4paper, 12pt]{elsarticle}
\usepackage{natbib}
\usepackage{amsthm}
\usepackage{amssymb}
\usepackage{graphicx,amsmath} 
\usepackage{tikz}
\usepackage{mathtools}
\usepackage{natbib}
\usepackage{booktabs,xcolor}
\usepackage{graphics}
\usepackage{enumitem}
\usepackage{subcaption}
\usepackage{caption}
\usepackage{setspace}
\usepackage[colorlinks=true,linkcolor=blue, urlcolor=blue, citecolor=blue, breaklinks]{hyperref}
\journal{the Journal of Mathematical Economics}

\newtheorem{proposition}{Proposition}

\theoremstyle{definition}

\newtheorem{definition}{Definition}

\newcommand\fnote[1]{\captionsetup{font=footnotesize}\caption*{#1}}
\usepackage{etoolbox}
\makeatletter
\patchcmd{\ps@pprintTitle}
  {Preprint submitted}
  {Preprint resubmitted}
  {}{}
\makeatother
\patchcmd{\emailauthor}{(#2)}{}{}{}
\patchcmd{\urlauthor}{(#2)}{}{}{}

\newcommand{\EE}{\mathbb{E}}

\newcommand{\rk}{\mathit{rk}} 
\begin{document}
\begin{frontmatter}

\title{Social Integration in Two-Sided Matching Markets}
\author{Josu\'e Ortega\corref{mycorrespondingauthor}\footnote{This paper is based on the third chapter of my PhD thesis. I am particularly indebted to Herv\'e Moulin for his advice and comments on previous drafts. I acknowledge helpful comments from two anonymous referees, Christopher Chambers, Julien Combe, Bram Driesen, Patrick Harless, Takashi Hayashi, Olivier Tercieux, Bumin Yenmez, and seminar participants at ZEW Mannheim. I am indebted to Sarah Fox, Christina Loennblad and Erin for proofreading the paper. Any errors are mine alone. A previous version of this paper appeared under the title ``Can everyone benefit from social integration?".}}
\ead{ortega@zew.de}
\address{Center for European Economic Research (ZEW), Mannheim, Germany.\\University of Essex, Colchester, UK.}

\begin{abstract}
When several two-sided matching markets merge into one, it is inevitable that some agents will become worse off if the matching mechanism used is stable. I formalize this observation by defining the property of integration monotonicity, which requires that every agent becomes better off after any number of matching markets merge. Integration monotonicity is also incompatible with the weaker efficiency property of Pareto optimality.

Nevertheless, I obtain two possibility results. First, stable matching mechanisms never hurt more than one half of the society after the integration of several matching markets occurs. Second, in random matching markets there are positive expected gains from integration for both sides of the market, which I quantify. 
\end{abstract}

\end{frontmatter}

\newpage
\setcounter{footnote}{0}
\section{Introduction}
\label{sec:introduction}
The (stable) marriage problem (MP), proposed by \cite{gale1962}, combines a rich mathematical structure with important real-life applications. The MP consists of two groups of agents, men and women. Each man wants to marry a woman, and each woman wants to marry a man. Both men and women have a strict preference ordering over the set of potential spouses. The marriage terminology is a metaphor for arbitrary two-sided matching problems, such as the assignment of workers to firms, students to universities, doctors to hospitals, and so on.

Gale and Shapley were interested in systematic rules to decide who should marry whom.  They referred to such rules as matching mechanisms. They recognized that the outcome of a matching mechanism should be {\it ``in accordance with agents' preferences''} and should satisfy {\it ``some agreed-upon criterion of fairness''}. Many such criteria can be formulated. Two standard criteria are stability and Pareto optimality. I refer to these requirements as efficiency criteria.

Other desirable criteria require that matching mechanisms incentivize the integration of different MPs ({\it communities}) into an extended MP ({\it the society}). I refer to these as pro-integration criteria. Matching mechanisms that satisfy pro-integration criteria would be particularly useful in practice because they would promote the operation of a large centralized clearinghouse instead of several small ones. For example, they would encourage hospitals to voluntarily enrol all patient-donor pairs in multi-hospital kidney exchange programs, instead of conducting kidney exchanges locally \citep{ashlagi2014,toulis2015}. They would also encourage charter and district-run schools to collaborate in the assignment of students to educational institutions \citep{manjunath2016,dogan2017,ekmekci2018}. Furthermore, they could ease the merger of school districts, which is an ongoing process in the US, given that 80\% of racial and ethnic segregation in its public schools occur between, but not within, school districts \citep{clotfelter2004,hafalir2017}.

In all the aforementioned cases, hospitals, schools, and school districts could increase the efficiency of their matching procedures by working together. Hospitals could increase the total number of successful kidney exchanges by using a multi-hospital clearinghouse, whereas schools and school districts could fill more seats using a centralized admission process that, at the same time, makes the final allocation of students to schools more diverse. Nevertheless, the aforementioned literature has documented that institutions are reluctant to participate in large centralized clearinghouses, partly because they anticipate that their new allocations after integration occurs will be worse than those they might obtain by matching within their own community only.

My first goal in this paper is to determine which efficiency criteria are compatible with pro-integration requirements. For this purpose, I define two intuitive pro-integration criteria. The first is {\it weak integration monotonicity (WIM)}. It requires that, whenever all disjoint communities merge as one, every agent weakly prefers the outcome obtained in an integrated society to the one obtained in a segregated one. For example, given a society with three communities $A$, $B$, and $D$, WIM requires that the matching obtained in society $A\cup B\cup D$ is weakly preferred by every agent to the one obtained in each community $A$, $B$, or $D$ alone. It is natural to expect that all communities will agree to integrate as a unified society whenever they use a matching mechanism that is WIM. 

Although WIM is a very intuitive pro-integration criterion, whether it is satisfied or not crucially depends on how we define the communities. Even though WIM may be satisfied for a society with communities $A$, $B$, and $D$, it may be violated when the communities become $\{A \cup B\}$ and $D$. If WIM is not satisfied in this latter case, should we expect the whole society to integrate, or should only communities $A$ and $B$ merge? Our second pro-integration property guarantees that integration occurs for every combination of communities. I call this property {\it integration monotonicity (IM)}. It ensures that complete social integration will always take place in a society, no matter the order in which the communities merge. 

In the previous example with three communities, IM requires that the matching obtained in societies $A\cup B$, $A \cup D$, and $B\cup D$ is weakly preferred by every agent to the one obtained in each society $A$, $B$, or $D$ alone. Furthermore, it also requires that the matching obtained in the completely integrated society $A \cup B \cup D$ is weakly preferred by everybody to that obtained in societies $A\cup B$, $A\cup D$, and $B\cup D$. Although IM is a stronger and less intuitive property than WIM, it is desirable because it guarantees that once full integration has been achieved, no subset of communities will want to separate from the rest and match only within themselves.\footnote{\cite{niederle2003} document this phenomenon in the matching of American gastroenterologists to hospitals. Even though local matching markets managed to merge into a centralized clearinghouse in 1986, some hospitals started defecting in 1994, until the centralized clearinghouse was completely abandoned in 2000.} In other words, whenever IM is violated, we can re-arrange the communities in a society so that WIM is also violated. 

Stability is incompatible even with WIM (Proposition \ref{thm:prop1}). This is not surprising given the well-known result establishing that adding a man to an MP makes every existing man weakly worse off. However, it is interesting that stable matching mechanisms never harm more than half of the society after the integration of communities occurs (Proposition \ref{thm:prop3}). This upper bound is tight, yet simulation results with random preferences suggest that the expected fraction of agents hurt by integration is considerably smaller, in fact consistently around one quarter of the society. It is also interesting that the much weaker efficiency property of Pareto optimality is incompatible with IM (Proposition \ref{thm:prop2}).

My second goal in this paper is to quantify the expected gains or losses from integration under stable matching mechanisms. I measure these welfare changes as the difference in the ranking of agents' spouses before and after integration occurs. For example, if I obtained my 5th best possible partner after integration occurred, but before this I obtained my 2nd best possible partner, there is a difference here of 3 positions, which shows that there are losses from integration. Societies may achieve social integration if the welfare losses, if any, are small compared to the size of the society, even in the absence of stable and pro-integration matching mechanisms.

I find that in large MPs, both men and women benefit from integration when agents' preferences over all potential spouses are drawn independently and uniformly at random ({\it a random MP}). In a random MP with $\kappa$ communities, each with $n$ men and $n$ women, the expected gains from integration for men and women under the men-optimal stable matching are asymptotically equivalent to $\log(n^{\kappa-1}/\kappa$) and $\kappa n \left( \frac{1}{\log n} - \frac{1}{\log \kappa n} \right)$, respectively. These expressions are derived in Proposition \ref{thm:prop4} and are good approximations even for MPs with only a few agents and communities. These approximations show that women particularly benefit from integration when the men-optimal stable matching is selected before and after integration occurs.

This paper proceeds as follows. Section \ref{sec:literature} discusses the related literature. Section \ref{sec:model} presents the model. Section \ref{sec:theory} establishes the compatibility between efficiency and pro-integration properties. Section \ref{sec:welfare} discusses the expected gains from integration. Finally, Section \ref{sec:conclusion} concludes.

\section{Related Literature}
\label{sec:literature}

This paper contributes to a recent body of literature studying the challenges of integration in different types of matching markets. \cite{manjunath2016}, \cite{dogan2017} and \cite{ekmekci2018} study the integration of different school admission processes. In their models, the set of students in two different communities is the same, i.e. communities are not disjoint. This implies that some students may receive two school offers, while some may get none. To ensure that all students receive at least one offer, \cite{manjunath2016} propose iteratively rematching students to schools. Although rematching may be costly, they show that a few iterations yield large welfare gains. Using a similar approach, \cite{dogan2017} show that students are weakly better off when all schools join a centralized clearinghouse. The previous two papers take as given that schools will merge their admission policies, and thus do not analyze the schools' incentives to integrate. 

In contrast, \cite{ekmekci2018} focus on the incentives of each school in joining a centralized clearinghouse. In their paper, a school has incentives to integrate if it gets better students participating in a centralized clearinghouse rather than running its own admission policy after all other schools have joined the aforementioned centralized clearinghouse. Using this sequential notion of integration, they obtain results similar to mine: no school has incentives to integrate. They consider restrictions on agents' preferences and particular policies that can make integration easier to achieve. 

\cite{hafalir2017} use the slightly different framework of matching with contracts to study the integration of school districts. A district is similar to a community in that it contains schools and students and in that all districts are disjoint. They show that if the choice functions of school districts respect initial endowments or a property they call {\it favoristic}, then integration under stable matching mechanisms does not harm a single student. Although this paper is close in spirit to theirs, they have different objectives. Their paper focuses on identifying conditions on schools' choice functions that guarantee that integration occurs, whereas this paper analyzes instead how integration can be achieved by computing the expected gains from integration in a random matching market. 

In the different context of kidney exchange, \cite{ashlagi2014} use a one-sided matching framework to study the incentives of hospitals to fully reveal their patient-donor pairs to a centralized clearinghouse. They call a matching mechanism individually rational for a hospital if it always matches more agents of that hospital after integration occurs. They show that in one-sided matching problems, there is no efficient matching mechanism that is maximal in the number of matchings created and, at the same time, is individually rational for every community. They and \cite{toulis2015} study mechanisms that are individually rational for multi-hospital kidney exchange, and show that in a large market, their proposed mechanism is asymptotically individually rational.

All the aforementioned results are inspired by the well-known monotonicity result which establishes that whenever a woman is added to a MP, every man becomes weakly better off. Similarly, adding an extra woman makes every existing woman weakly worse off. This result extends to many-to-one MPs (Theorem 5 in \citealp{kelso1982}, Theorems 1 and 2 in \citealp{crawford1991}, and Theorems 2.25 and 2.26 in \citealp{roth1992}). \citet{ashlagi2017} describe the exact magnitude of these welfare changes in random MPs. \cite{toda2006} uses this monotonicity property to characterize the set of stable outcomes. 

Related notions of integration monotonicity have been explored in cooperative game theory. In particular, \cite{sprumont1990} defines {\it population monotonicity}, a more general property that allows groups of arbitrary size to integrate. It requires that whenever two communities of arbitrary size join as one, there exists a way of sharing the surplus generated by its members so that every agent becomes weakly better off. He provides a characterization of population monotonicity using monotonic games with veto players. Sprumont's work only applies to games with transferable utility. This large class of games does not include the MP. My notion of matching schemes is inspired by his notion of assignment schemes.

More recently, \cite{chambers2017} use the population monotonicity axiom (which they call integration monotonicity) to study exchange economies. The definition of integration monotonicity that I use corresponds to theirs. They find that an allocation rule satisfies Pareto optimality and integration monotonicity only if the order in which communities integrate matters. Interestingly, they show that whenever integration occurs, it must hurt at least a third of the society if equals are treated equally.

The term population monotonicity is sometimes used to refer to the solidarity axiom, introduced by \cite{thomson1983}. This axiom is suitable for scenarios in which a society produces a fixed amount of welfare, no matter how many agents belong to it. Examples of those include bargaining and fair division. The solidarity axiom requires that the burden imposed by a new agent should be shared by all existing members of a society. Population monotonicity and solidarity are similar concepts but have very different interpretations.

\section{Model}
\label{sec:model}

A {\it community} $C$ is a set of men and women, and contains at least one person of each gender. There are $\kappa$ disjoint communities. Given a set of communities, the {\it society} $S$ is the set of all communities. A {\it population} $P \subseteq S$ is a set of communities in $S$. $M^P$ and $W^P$ denote the sets of men and women in each community $C \in P$. $C_x$ denotes the community to which person $x$ belongs. As it will cause no confusion, I will use the notation $x \in P$ whenever $\exists \, C \in P$ such that $x \in C$, i.e. whenever a person $x$ belongs to a community in population $P$.

Each man $m$ (resp. woman $w$) has strict preferences over the set of all women in the society $W^S$ (resp. men $M^S$).\footnote{For simplicity, I assume that each man (resp. woman) prefers matching with any woman (resp. man) rather than remaining alone.} I write $w \succ_m  w'$ to denote that $m$ prefers $w$ to $w'$. Similarly, $w \succcurlyeq_m w'$ if either $w \succ_m w'$ or $w = w'$. I represent women's preferences using the same notation. Given a population $P$, I call $\succ_P \coloneqq (\succ_x)_{x \in (M^P \cup W^P)}$ a {\it preference profile} of the agents belonging to the communities in $P$. An {\it extended MP} is a triple $(M^S,W^S;\succ_S)$. 

Given a population $P$, a {\it matching} $\mu:(W^P \cup M^P) \mapsto (W^P \cup M^P$) is a function of order two ($\mu^2(x)=x)$ such that if $\mu(m) \neq m$, then $\mu(m) \in W^P$, and if $\mu(w) \neq w$, then $w \in M^P$. An agent is matched to himself if he remains unmatched. A {\it matching scheme} $\sigma: (M^S \cup W^S) \times 2^S \mapsto (M^S \cup W^S)$ is a function that specifies a matching $\sigma(\cdot,P)$ for every $P \in 2^S$, i.e. $\sigma (\cdot,P): (M^P \cup W^P) \mapsto (M^P \cup W^P)$. A {\it matching mechanism} $\Gamma$ is a function $\succ \mapsto \Gamma(\succ)$ from the set of all preference profiles to the set of all matching schemes.

To give an example of an extended MP and a matching scheme, consider a society with two communities, $A$ and $B$, each with one man and one woman. Agents' preferences are such that the man from community $A$ ($m^A$) prefers the woman from community $A$ ($w^A$) over the woman from community $B$ ($w^B$). The preferences of the remaining agents are summarized below
\begin{eqnarray*}
w^A \succ_{m^A} w^B &\quad& m^B \succ_{w^A} m^A\\
w^A \succ_{m^B} w^B &\quad& m^B \succ_{w^B} m^A
\end{eqnarray*}
A possible matching scheme $\sigma$ is $\sigma(m^A,A)=w^A$, $\sigma(m^B,B)=w^B$, $\sigma(m^A,S)=w^B$ and $\sigma(m^B,S)=w^A$. We will return to this example in Proposition \ref{thm:prop1}.

\subsection{Efficient Matching Schemes}

I consider two well-known efficiency properties. The first is stability. Besides its intuitive appeal, the concept of stability is a good predictor of the success of several real-life matching mechanisms \citep{roth1992}.  

\begin{definition}[Stability] 
A matching $\mu:(W^P \cup M^P) \mapsto (W^P \cup M^P)$ is stable if there is no man $m \in M^P$ and no woman $w \in W^P$ that are not married to each other $(\mu(m) \neq w)$ such that $w  \succ_m \mu(m)$ and $m \succ_w \mu(w)$. Any such pair ($m,w$) is called a blocking pair.
\end{definition}

A weaker efficiency property is Pareto optimality. It is arguably the most basic efficiency consideration in economics. It only requires that there is no way of making one agent better off without hurting any other agent.

\begin{definition}[Pareto optimality] 
A matching $\mu:(W^P \cup M^P) \mapsto (W^P \cup M^P)$ is Pareto optimal if there is no other matching $\mu'$ such that $\mu'(x) \succcurlyeq_x \mu(x)$ for every agent $x \in (W^P \cup M^P)$, and $\mu'(y) \succ_y \mu(y)$ for some agent $y \in (W^P \cup M^P)$.
\end{definition}

The properties of matchings trivially extend to matching schemes and mechanisms. A matching scheme $\sigma$ is stable (resp. Pareto optimal) if the matching $\sigma(\cdot,P)$ is stable (resp. Pareto optimal) in every population $P \subseteq S$. A matching mechanism $\Gamma$ is stable (resp. Pareto optimal) if the matching scheme $\Gamma(\succ)$ is stable (resp. Pareto optimal) with respect to the preference profile $\succ$.

\subsection{Pro-Integration Matching Schemes}
I introduce two pro-integration properties. The first is weak integration monotonicity (WIM). It requires that no agent is hurt whenever all communities integrate.

\begin{definition}[Weak Integration Monotonicity]
A matching scheme $\sigma$ is WIM if $\forall C \in S$ and $\forall x \in (M^C \cup W^C$), $\sigma (x,S) \succcurlyeq_x \sigma(x,C_x)$.
\end{definition}

WIM matching schemes ensure that every agent is weakly better off after communities have merged, and thus it encourages the complete integration of disjoint communities.\footnote{WIM is stronger requirement than the individual rationality property in \cite{ashlagi2014}, which requires that the number of agents matched from each community weakly increases after the integration of all communities. Individual rationality is satisfied if all agents marry the worst possible partner after integration occurs. The number of marriages crated after integration is (weakly) higher, but (some) agents are clearly worse off.}

A stronger concept is integration monotonicity (IM). It requires that no agent is hurt, whenever any two disjoint populations integrate.

\begin{definition}[Integration Monotonicity]
A matching scheme $\sigma$ is IM if $\forall P,P' \in 2^S$ such that $P \cap P' = \emptyset$ and $\forall x \in (M^P \cup W^P)$, $\sigma(x,P \cup P') \succcurlyeq_x \sigma(x,P)$.
\end{definition}

A matching mechanism $\Gamma$ is IM (resp. WIM) if the matching scheme $\Gamma(\succ)$ is IM (resp. WIM) with respect to the preference profile $\succ$. 

An IM matching scheme guarantees two important properties. First, that complete social integration will occur independently of the order in which we merge the communities (the order may affect who marries whom, but integration always makes agents better off with respect to matching only within a sub-population of the society). Second, that after the integration of all communities has taken place, there is no subset of communities that wish to separate from the rest of the society to match within themselves instead. These two desiderata justify the definition of IM.

Before presenting my results, I discuss the relationship between efficiency and pro-integration properties. One may conjecture that integration monotonicity combined with Pareto optimality implies stability. This conjecture is natural given similar existing results in the literature.\footnote{In cooperative games with transferable utility, population monotonicity combined with efficiency implies the core property \citep{sprumont1990}. In exchange economies, integration monotonicity combined with efficiency implies the core property \citep{chambers2017}.} However, this conjecture is false. Consider a society with two communities. One has a pretty man and an ugly woman, whereas the other one has an ugly man and a pretty woman. The matching scheme that always marries people to those in their own community is integration monotonic and Pareto optimal (if ugly people prefer pretty over ugly people). However, this matching scheme is not stable: the pretty man and woman prefer each other to their spouses.

\section{Results}
\label{sec:theory}
It would be ideal if we could construct a stable and IM matching mechanism. Unfortunately, we lack even a stable and WIM matching mechanism. 

\begin{proposition}
\label{thm:prop1}
For each society with at least two communities, no stable matching mechanism is weakly integration monotonic.
\end{proposition}

\begin{proof} Let $A$ and $B$ be the two communities, each with one man and one woman. Consider the following preference profile $\succ$
\begin{eqnarray*}
w^A \succ_{m^A} w^B &\quad& m^B \succ_{w^A} m^A\\
w^A \succ_{m^B} w^B &\quad& m^B \succ_{w^B} m^A
\end{eqnarray*}

All women prefer $m^B$ and all men prefer $w^A$. Stability of a matching scheme requires that $\sigma(w^A,A)=m^A$, $\sigma(w^B,B)=m^B$, and $\sigma(w^A,A \cup B)=m^B$. However, both man $m^A$ and woman $w^B$ obtain a worse partner when communities $A$ and $B$ merge. Therefore, for the preference profile $\succ$, any stable matching mechanism $\Gamma$ produces a matching scheme $\Gamma(\succ)$ that is not WIM.
\end{proof}

Proposition \ref{thm:prop1} is an expected result, given the well-known result that states that adding a man makes every woman weakly worse off, discussed in the Introduction. However, Proposition \ref{thm:prop3} shows that, although stable matching mechanisms cannot guarantee that all agents are weakly better off after integration occurs, they never harm more than one half of the agents in the society.

\begin{proposition}
\label{thm:prop3}
In any stable matching scheme $\sigma^*$,
\begin{equation*}
\left\vert \{x \in (M^S \cup W^S) : \sigma^*(x,C_x) \succ_x \sigma^*(x,S)\} \right\vert \leq \frac{\left\vert M^S \cup W^S \right\vert}{2}
\end{equation*}
The bound is tight.
\end{proposition}

\begin{proof}
Let us partition $S$ into three sets $S^0$, $S^+$ and $S^-$, defined as
\begin{eqnarray*}
S^0&\coloneqq&\{ x \in (M^S \cup W^S) : \sigma^*(x,S) =\sigma^*(x,C_x) \}\\
S^+&\coloneqq&\{ x \in (M^S \cup W^S) :  \sigma^*(x,S) \succ_x \sigma^*(x,C_x)\}\\
S^-&\coloneqq&\{ x \in (M^S \cup W^S) :  \sigma^*(x,C_x) \succ_x \sigma^*(x,S)\}
\end{eqnarray*}

so that $S^0$ is the set of people who keep the same partner after integration, $S^+$ are those who prefer their ``integrated'' partner, and $S^-$ are those who prefer their ``segregated'' partner.

Consider an arbitrary couple ($x,\sigma^*(x,C_x)$). If $x \in S^-$, $\sigma^*(x,C_x) \in S^+$ because otherwise ($x,\sigma^*(x,C_x)$) constitutes a blocking pair to the matching $\sigma^*(\cdot, S)$, contradicting the fact that $\sigma^*$ is a stable matching scheme. It follows that $\left\vert S^+ \right\vert \geq \left\vert S^- \right\vert$, and thus $\left\vert S^0 \right\vert + \left\vert S^+ \right\vert \geq \left\vert S^- \right\vert$. It follows that $\left\vert \{x \in (M^S \cup W^S) : \sigma^*(x,S) \succcurlyeq_x \sigma^*(x,C_x)\} \right\vert \geq \left\vert M^S \cup W^S \right\vert/2$, completing the proof. The example used in the proof of Proposition \ref{thm:prop1} shows that the one-half bound is tight.
\end{proof}

Proposition \ref{thm:prop1} shows that stability cannot be combined with pro-integration criteria. However, Proposition \ref{thm:prop2} shows that even the considerably weaker efficiency property of Pareto optimality is also incompatible with IM. 

\begin{proposition}
\label{thm:prop2}
For each society with at least three communities, no Pareto optimal matching mechanism is integration monotonic.
\end{proposition}

\begin{proof} Let $A$, $B$ and $D$ be the three communities, each with one man and one woman. Consider the following preference profile $\succ$
\begin{eqnarray*}
w^B \succ_{m^A}  w^D \succ_{m^A} w^A &\quad& m^B  \succ_{w^A}  m^D  \succ_{w^A}  m^A\\
w^D  \succ_{m^B}  w^A  \succ_{m^B}  w^B &\quad&  m^D  \succ_{w^B}  m^A  \succ_{w^B}  m^B\\
w^A  \succ_{m^D}  w^B  \succ_{m^D}  w^D &\quad&  m^A  \succ_{w^D}  m^B  \succ_{w^D}  m^D
\end{eqnarray*}

Agents' preferences are such that agents of community $A$ prefer those from $B$, agents from $B$ prefer those from $D$, and agents from $D$ prefer to those from $A$. All Pareto optimal matching schemes require $\sigma(w^A,A\cup B)=m^B$, $\sigma(w^B,B\cup D)=m^D$, and $\sigma(w^D,D\cup A)=m^A$. Note that in any Pareto optimal matching scheme, there is always a community that gets its first choice whenever we merge only two communities.

IM requires that when we aggregate all communities, all agents should do at least as well as when only two societies merge, namely
\begin{eqnarray*}
\sigma(m^A,S) \succcurlyeq_{m^A} w^B && \sigma(w^A,S) \succcurlyeq_{w^A} m^B\\
\sigma(m^B,S) \succcurlyeq_{m^B} w^D && \sigma(w^B,S) \succcurlyeq_{w^B} m^D\\
\sigma(m^D,S) \succcurlyeq_{m^D} w^A && \sigma(w^D,S) \succcurlyeq_{w^D} m^A
\end{eqnarray*}

This is impossible, because some agent would no longer be able to obtain her first choice. 
\end{proof}

Propositions \ref{thm:prop1} and \ref{thm:prop2} show that efficiency and pro-integration properties are at odds with each other.\footnote{IM and Pareto optimality are compatible in a society that only has two communities. This is because, for such a society, IM and WIM are equivalent.} They can only be satisfied together in their weak versions. It is obvious that {\it matching schemes that are Pareto optimal and WIM always exist.} We can obtain them by marrying each agent within their own community, and then repeatedly perform Pareto improvements on that initial matching. 

Note that populations are sets of exogenously given sets of communities. One could consider a stronger notion of integration monotonicity which requires weak improvements every time two arbitrary sets of agents (which are not sets of communities) merge. This notion is stronger than the one I consider. Consequently, the impossibility described in Proposition \ref{thm:prop2} carries over when using this stronger notion of integration monotonicity.

\section{Gains from Integration in Stable Matching Mechanisms}
\label{sec:welfare}

In this section, I quantify the welfare changes that occur after integration has taken place. I measure those changes by the difference in the rank of the spouse that agents obtain before and after integration. Because this analysis uses existing results, henceforth I will assume that each community has $n$ men and $n$ women. Albeit restrictive, this assumption allows me to isolate the effects of integration from those of having societies that are unbalanced in their gender ratio, which are discussed in detail by \cite{ashlagi2017}. This assumption implies that all agents are married in any stable matching.\footnote{I discuss the integration of unbalanced communities in Section \ref{sec:conclusion}.}

I also assume throughout that the matching that occurs before and after integration is the men-optimal stable matching (MOSM). The MOSM is a stable matching such that no man gets a better partner in any other stable matching, and it always exists.  This assumption allows me to compare agents' welfare across different stable matchings, and is also imposed by \citet{dogan2017}, \citet{hafalir2017} and \citet{ekmekci2018}. This assumption is justified in several applications. For example, the student-optimal stable matching is appealing and consistently selected in school choice \citep{abdulkadirouglu2003}. Similarly, the resident-optimal stable matching has been used to allocate medical interns to hospitals in the U.S. \citep{roth1999}. I denote by $\sigma^*$ the men-optimal matching scheme, so that $\sigma^*(\cdot, P)$ is the MOSM for the matching problem $(M^P,W^P; \succ_P)$ and $\sigma^*(\cdot,C)$ denotes the MOSM for the MP $(M^{C_i},W^{C_i},\succ_{C_i}),$ $\forall C_i \in S$.\footnote{Formally, I should write $\sigma^*_{\succ_P}$ because the MOSM depends upon the preference profile $\succ_P$. Similarly, I should write below $\rk_m(w, \succ_m)$ and $\rk_M (\mu, \succ_M)$ to denote a man's rank of women and men's average rank of wives, because they both depend on the preference ordering $\succ_m$ and preference profile $\succ_M$. However, I simplify the notation following the literature \citep[][p. 75]{ashlagi2017}, as it will be clear that these objects depend on agents' preferences.}

Now I introduce some useful definitions. Given a society, an {\it extended random MP} is generated by drawing a complete preference list for each man and each woman independently and uniformly at random. Random MPs were first studied by \cite{wilson1972}, and have been widely studied ever since. The {\it absolute rank} of a woman $w$ in the preference order of a man $m$ (over all potential spouses in the society) is defined by $\rk_m(w) \coloneqq \left\vert \{ w' \in W^S: w' \succcurlyeq_m w \} \right\vert$. Similarly, I use $\rk_w(m)$ to denote the absolute rank of $m$ in the preference order of $w$. Given a matching $\mu$, the {\it men's absolute average rank of wives} is defined by
\begin{eqnarray*}
\rk_M (\mu) &\coloneqq& \frac{1}{\left\vert M^S \right\vert} \sum_{m \in M^S} \rk_m(\mu(m))
\end{eqnarray*}

The same notation is used to denote the women's average rank of husbands $\rk_W(\mu)$. The gains from integration for men $\gamma_M$ represent the difference between the men's average rank of wives that they obtain before and after integration occurs under the MOSM. Formally,
\begin{eqnarray*}
\gamma_M &\coloneqq& \rk_M(\sigma^*(\cdot,C)) - \rk_M(\sigma^*(\cdot,S))
\end{eqnarray*}

The same notation is used to denote women's gains from integration $\gamma_W$. A higher ranking means a less desired person and thus, there are expected gains from integration for men and women whenever $\gamma_M$ and $\gamma_W$ are positive, respectively. Proposition \ref{thm:prop4} establishes that the expected gains from integration are positive for both men and women, and provides an approximation\footnote{Functions $f(n,\kappa)$ and $g(n,\kappa)$ are asymptotically equivalent $f(n,\kappa) \sim g(n,\kappa)$ if $\lim_{n,\kappa \to \infty} \frac{f(n,\kappa)}{g(n,\kappa)}=1$.} for MPs in which the number of agents and communities is large.

\begin{proposition}
\label{thm:prop4}
In a random extended MP with $\kappa$ communities, each with $n$ men and $n$ women,
\begin{eqnarray}
\label{eq:gainsm} \EE [\gamma_M] &\sim& \log\left(\frac{n^{\kappa-1}}{\kappa}\right)\\
\label{eq:gainsw} \EE[\gamma_W] &\sim& \kappa n \left( \frac{1}{\log n}-\frac{1}{\log \kappa n} \right)
\end{eqnarray}
\end{proposition}

\begin{proof}
I start the proof with some preliminaries regarding standard MPs with $\kappa=1$. The \citet{mcvitie1971} algorithm is a modification of the deferred acceptance algorithm in which only one man proposes to a woman at each step. Let $X$ denote the number of proposals required by this algorithm to find the MOSM. Using the coupon collector problem, we obtain (see \citealp{motwani1995}; p. 58)
\begin{equation*}
\EE[X] = n \log n +O(n) \sim n \log n
\end{equation*}

Because men propose first to their most preferred women, $\EE(X)/n$ is the expected men's absolute average rank of wives, so 
\begin{equation*}
\EE[\rk_M(\sigma^*(\cdot,S))] = \log n + O(1) \sim \log n
\end{equation*}

Each woman receives $\log(n)$ proposals in expectation. The one she accepts is the one coming from her most preferred partner. The random variable $[\rk_W(\sigma^*(\cdot,S))-1]$ is binomially distributed with parameters $(n-\log n, \frac{1}{\log n+ 1})$. Therefore, 
\begin{equation*}
\EE[\rk_W(\sigma^*(\cdot,S))] = \frac{n}{\log n + 1} + O(1) \sim \frac{n}{\log n } 
\end{equation*}

These results were proven by \cite{wilson1972}, \citet{knuth1997} and \citet{pittel1989},\footnote{\cite{pittel1989} proves a stronger statement, namely that $\frac{\rk_M(\sigma^*(\cdot,S))}{\log n} \overset{p}{\to} 1$ and $\frac{\rk_W(\sigma^*(\cdot,S))}{n/\log n} \overset{p}{\to} 1$.} and are explained in detail by \cite{canny2001}. They imply that in an extended MP with $\kappa$ communities
\begin{eqnarray*}
\EE[\rk_M(\sigma^*(\cdot,S))] &\sim& \log \kappa n \\
\EE[\rk_W(\sigma^*(\cdot,S))] &\sim& \frac{\kappa n}{ \log \kappa n}
\end{eqnarray*}

The {\it relative rank} of a woman $w$ in the preference order of a man $m$ is defined by $\widehat{ \rk}_m(w) \coloneqq \left\vert \{ w' \in C_m : w' \succcurlyeq_m w \} \right\vert$. It indicates the ranking of woman $w$ with respect to all women who belong to the same community as $m$. Similarly, I denote the relative rank of $m$ in the preference order of $w$ by $\widehat{ \rk}_w(m)$. Given a matching $\mu$, the {\it men's relative average rank} of wives is defined by
\begin{eqnarray*}
\widehat{ \rk}_M (\mu) &\coloneqq& \frac{1}{\left\vert M^S \right\vert} \sum_{m \in M^S} \widehat{ \rk}_m(\mu(m))
\end{eqnarray*}

The same notation is used to denote the women's relative average rank of husbands $\widehat{ \rk}_W$. Pittel's result also implies that 
\begin{eqnarray*}
\EE[\widehat{ \rk}_M(\sigma^*(\cdot,C))] &\sim& \log n \\
\EE[\widehat{ \rk}_W(\sigma^*(\cdot,C))] &\sim& \frac{n}{ \log n}
\end{eqnarray*}

So that before integration, men and women get a partner {\it relatively} ranked $\sim n$ and $\sim \log n$ inside their own community. We need to figure out in which position these partners are in the absolute ranking of all $M^S$ and $W^S$. To answer this question, suppose that an agent is ranked $q$ among all agents in his community. A random agent from another community could be better ranked than agent 1, between agents 1 and 2, ..., between agents $q-1$ and $q$, between agents $q$ and $q+1$, and so on. Therefore, a random agent from another community is in any of those gaps with probability $1/(n+1)$ and thus has $q/(n+1)$ chances of being more highly ranked than our original agent with the relative ranking $q$. There are $n(\kappa -1)$ men from other communities. On average, $\frac{qn(\kappa-1)}{n+1}$ men will be ranked better than him. Furthermore, there were already $q$ men in his own community ranked better than him. This implies that his expected ranking is $q+\frac{qn(\kappa-1)}{n+1} = \frac{q(\kappa n+1)}{n+1} \sim q\kappa$. Substituting $q$ for $\log(n)$ and $n/\log(n)$, respectively, we obtain the expected average rank of wives and husbands before integration. It follows that 
\begin{eqnarray*}
\EE [\gamma_M] &=&\log(n)\kappa-\log(\kappa n) \sim \log\left(\frac{n^{\kappa-1}}{\kappa}\right)\\
\EE[\gamma_W] &=& \frac{n}{\log(n)} \kappa-\frac{\kappa n}{\log(\kappa n)} \sim \kappa n \left(\frac{1}{\log n}- \frac{1}{\log \kappa n} \right)
\end{eqnarray*}
\end{proof}

Admittedly, we require strong assumptions to obtain Proposition \ref{thm:prop4}. We need balanced and equally populated communities, as well as independently and uniformly distributed preferences. Although those assumptions are likely to be violated in practice, Proposition \ref{thm:prop4} sheds some light on how the most important stable matching mechanism (the MOSM) may promote social integration, despite its incompatibility with WIM. Not only does this matching mechanism never hurt more than one half of society (Proposition \ref{thm:prop3}), but it also improves the expected average ranking of men's wives and women's husbands after integration occurs (under the aforementioned assumptions). 

Although Proposition \ref{thm:prop4} is an asymptotic result, Figure \ref{fig:gains} shows that the expected gains from integration for large societies described in Proposition \ref{thm:prop4} are good approximations even for small values of $n$ and $\kappa$. Note that women benefit more from integration. The intuition behind this observation is that women obtain a partner who is ranked relatively 	very poorly before integration occurs under the MOSM. This poor relative rank is amplified when compared to the absolute rank in their complete preference order over all partners. 

\begin{figure}[htbp!]
 \caption{Gains from integration, by gender.}
	\label{fig:gains}
    \centering
    \begin{subfigure}[b]{0.475\textwidth}
		\centering
        \includegraphics[width=\textwidth]{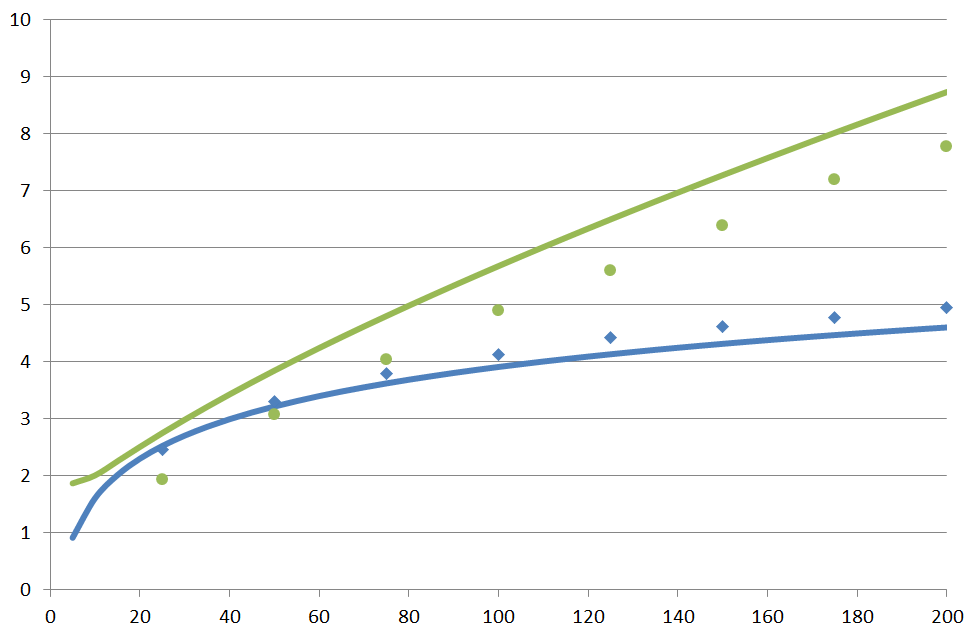}
        \caption{$r=2$}
        \label{fig:r2}
    \end{subfigure}
    ~ 
    \begin{subfigure}[b]{0.475\textwidth}
        \includegraphics[width=\textwidth]{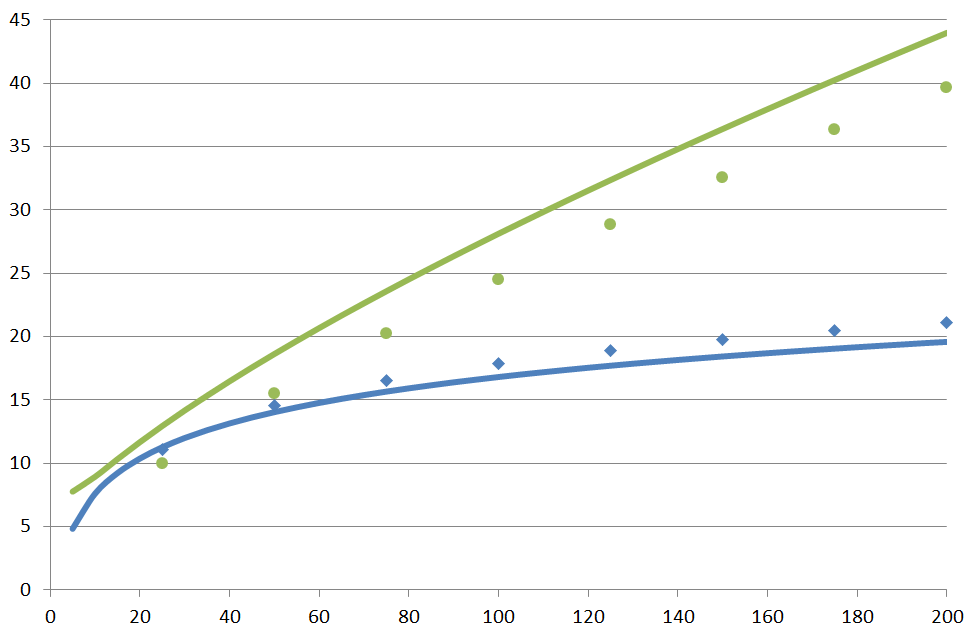}
        \caption{$r=5$}
        \label{fig:r5}
    \end{subfigure}
		\fnote{The lines denote the asymptotic expressions \ref{eq:gainsm} and \ref{eq:gainsw}. The dots represent the average gains from 10,000 simulations. Horizontal axis: $n$, vertical axis: (expected) gains from integration. Blue observations correspond to men, whereas green ones correspond to women.}
\end{figure}

Before concluding the paper, I describe some characteristics regarding those agents who are hurt by integration. These observations come from simulation exercises described in detail in the Appendix. First, the number of women hurt by integration under the MOSM is larger than the corresponding number of men (women 74\%, men 26\%  with $n=500,\kappa=5$). Second, the magnitude of their losses is very different. Women are severely hurt (average loss of 303 ranking places with $n=500,\kappa=5$), whereas men suffer a moderate loss at most (average loss of 8 ranking places). Women's losses become smaller when their preferences are correlated. 

The magnitude of the losses of agents hurt by integration becomes smaller compared to the size of the problem for both men and women. However, I find it very interesting that the fraction of society hurt by integration remains relatively constant around 25\%, even when preferences are correlated. 

\section{Conclusion}
\label{sec:conclusion}

Although it is impossible to guarantee that the integration of disjoint matching markets will benefit everyone, the integration process never hurts more than half of the agents if a stable matching is systematically chosen. Furthermore, there are positive expected gains from integration in random matching markets, which are larger for the side of the market that receives the proposals in the deferred acceptance algorithm.

Two extensions of this paper are:

{\it 1. Many-to-one Matching.} The results presented in Section \ref{sec:model} extend to the many-to-one case when agents have responsive preferences. In this case, every many-to-one problem has an equivalent formulation where, instead of assigning students to colleges, students are assigned to college seats \citep{roth1992}. Proposition \ref{thm:prop3}, which establishes that no more than half of the agents are hurt after integration occurs, is now interpreted as half of colleges' seats and students. In a many-to-one framework, integration can hurt more than half of the students and colleges, even when preferences are responsive. 

It is difficult to approximate the gains from integration in many-to-one matching problems, because the classical results from Wilson, Knuth, and Pittel that I use in the proof of Proposition \ref{thm:prop4} only apply to one-to-one matching problems.

{\it 2. Unbalanced Societies.} All results in Section \ref{sec:model} extend to the case of societies that are	 unbalanced in their gender ratio. Regarding the expected gains from integration, \cite{ashlagi2017} show that the spouse rank for men and women in the MOSM roughly reverses if there are more men than women. Therefore, if we merge several societies, all in which there are more men than women, the expected gains from integration roughly reverse as well. 

What happens if we merge some communities with more men than women, and some with more women than men? It is significantly harder to compute the gains from integration, which depend on the size of the gender imbalance in each society. However, if a male-dominated community (with more men than women) merges with a female-dominated one, men from the male-dominated community experience larger expected gains from integration than men from the female-dominated community. Similarly, women from the female-dominated community obtain larger gains from integration, as compared to women from the male-dominated community. 

I leave a deeper analysis of integration of many-to-one and unbalanced matching markets for future research.

\section*{Appendix}

I present the detailed simulations discussed in Section \ref{sec:conclusion}. In all cases, the results presented are averages over 1,000 simulations. The numbers in parentheses indicate standard errors. The corresponding code is available at \href{www.josueortega.com}{www.josueortega.com}. The simulations were run at the High Performance Computing facilities at the University of Glasgow. I present the simulation results using three tables.

\begin{table}[!ht]
\begin{center}
\caption[caption,justification=centering]{Percentage of the society that prefers segregation}
\label{tab:tab1}
\begin{tabular}{lccc}
\hline
\hline
$\kappa \backslash n$  & 50 & 100& 500\\
\hline
2 & 25.4 {\small [40,60]} & 25.8 {\small [38,62]} &25.5 {\small [36,64]}\\
  &  \footnotesize  (0.04) & \footnotesize  (0.02) & \footnotesize  (0.01)\\
3 & 25.4 {\small [35,65]}& 25.8 {\small [34,66]}& 26 {\small [30,70]}\\
 & \footnotesize  (0.03) & \footnotesize  (0.01) & \footnotesize  (0.01)\\
4 & 24.8 {\small [32,68]} & 25.1{\small [29,71]}&25.6 {\small [26,74]}\\
 &  \footnotesize  (0.02) &  \footnotesize  (0.01) & \footnotesize  (0.01)\\
5 & 24.3 {\small [28,72]} & 24.6 {\small [26,74]}& 25.1 {\small [26,74]} \\
 & \footnotesize  (0.02) & \footnotesize  (0.01) & \footnotesize  (0.01)\\
\hline
\hline
\end{tabular}
\end{center}
\end{table}

Table \ref{tab:tab1} describes the fraction of the society that is hurt by integration for several combinations of parameters. The division of this group between men and women appears in brackets, i.e. for $n=50$ and $\kappa=2$, 25.4\% of society becomes worse off after integration occurs. 40\% of these 25.4\% of society are men, whereas the remaining 60\% are women. I find it surprising that the fraction of the society that gets hurt by integration remains relatively constant around 25\%. It is also interesting that in larger matching problems, the fraction of men hurt by integration becomes smaller.

\begin{table}[!ht]
\centering
\caption[caption,justification=centering]{Average welfare loss experienced by people who prefer segregation, by gender.}
\label{tab:welfareloss}
\begin{tabular}{lrlrlrl}
\hline
\hline
$\kappa \backslash n$ & \multicolumn{2}{c}{50} & \multicolumn{2}{c}{100}& \multicolumn{2}{c}{500}\\
& men & women & men & women & men & women \\
\hline
2 & 4.9	&19.7&	5.7&	34.9& 7.4 & 136.6 \\
 &  \footnotesize (0.91)  &  \footnotesize (9.47) & \footnotesize (0.95)  & \footnotesize (22.91) & \footnotesize (1.04) & \footnotesize (246.83)\\
3 & 5.4	&27.4&	6.2&	49.2&7.9&193.8\\
 & \footnotesize (1.01)  & \footnotesize (14.57) & \footnotesize (0.96)  & \footnotesize (39.22) & \footnotesize (0.83) & \footnotesize (409.86)\\
4 & 5.7	&35&	6.5&	62.1&  8.1 &250.8\\
 & \footnotesize (1.08)  & \footnotesize (21.67) & \footnotesize (1.03)  & \footnotesize (60.48) & \footnotesize (0.84) & \footnotesize (623.95)\\
5 & 6	&41.8&	6.8&	74.9&8.4&303.3\\
 & \footnotesize (1.07) & \footnotesize (28.26) & \footnotesize (1.14)  & \footnotesize (84.86) & \footnotesize (0.88) & \footnotesize (709.15)\\
\hline
\hline
\end{tabular}
\end{table}

Table \ref{tab:welfareloss} describes the average losses experienced by those hurt by integration after all communities merge. Men suffer very moderate welfare losses with the MOSM, whereas women who are hurt by integration suffer large reductions in the ranking of their spouses. Interestingly, for both men and women, the welfare losses divided by the size of the problem $\kappa n$ seem to converge to zero. This observation suggests that in large markets, the welfare losses eventually become negligible.

\begin{table}[!htbp]
\begin{center}
\caption[caption,justification=centering]{Statistics for correlated preferences, $n=100$, $\kappa=2$.}
\label{tab:corr}
\begin{tabular}{lcrlrl}
\hline
\hline
$\rho$ & Fraction   & \multicolumn{2}{c}{Exp. ranking}& \multicolumn{2}{c}{Welfare loss}\\
	& worse off & men & women  & men & women         \\
		& (1)  & \multicolumn{2}{c}{(2)}& \multicolumn{2}{c}{(3)}\\

	\hline
0.9   & 24.6       & 105.6            & 93.9          & 30.5            & 31.9 \\
      & \footnotesize (0.02)       & \footnotesize (59.22)            & \footnotesize (50.78)          & \footnotesize (17.75)            & \footnotesize (24.36) \\
0.7   & 26.1        & 114.3           & 89.2          & 17.2            & 34.7\\
      & \footnotesize (0.02)       & \footnotesize (25.51)            & \footnotesize (20.07)          & \footnotesize (8.45)             & \footnotesize (25.87) \\
0.5   & 26.1        & 109.7           & 93.4          & 10.9             & 34.8 \\
      & \footnotesize (0.02)       & \footnotesize (15.81)             & \footnotesize (11.85)          & \footnotesize (3.46)             & \footnotesize (22.99)  \\
0.3   & 25.9        & 104.4             & 97.5           & 7.8             & 34.1 \\
      & \footnotesize (0.02)       & \footnotesize (6.46)             & \footnotesize (4.43)           & \footnotesize (1.83)              & \footnotesize (21.68) \\
0.1   & 25.8        & 101.1           & 100.1         & 6.1              & 34.1 \\
      & \footnotesize (0.02)       & \footnotesize (1.05)             & \footnotesize (0.62)           & \footnotesize (1.24)             & \footnotesize (22.35) \\			
\hline
0   & 25.8       & 100.6           & 100.5         & 5.7              & 34.9 \\
      & \footnotesize (0.02)       & \footnotesize (0.22)             & \footnotesize (0.1)           & \footnotesize (0.95)             & \footnotesize (22.91) \\
\hline
\hline
\end{tabular}
\end{center}
\end{table}

Finally, Table \ref{tab:corr} computes the same variables described in Tables \ref{tab:tab1} and \ref{tab:welfareloss} but imposing correlation in agents' preferences. Correlated preferences are evident in some matching environments like school choice. I introduce correlation in preferences as follows. I define a status quo in preferences for both men and women. The status quo is a random order over all possible partners. Each agent's preferences are identical to the status quo, except perhaps in at most $c$ positions. For example, if $c=2$ and the status quo over six partners is $1\succ 2\succ 3\succ 4\succ 5\succ 6$, an agent's preferences could be $1\succ 2\succ 6\succ 4\succ 5\succ 3$ but not $2\succ 3\succ 1\succ 4 \succ 5 \succ 6$. The swaps in agents' preferences are chosen randomly. The expected correlation coefficient between each agents' preferences and the status quo equals $\rho=1-\frac{c}{\kappa n}$. The simulation results using correlated preferences appear in Table \ref{tab:corr}.\footnote{There are several ways of introducing correlation in preferences, e.g. \cite{caldarelli2001}, \cite{boudreau2010}, \cite{hafalir2013}, \cite{lee2017}, and \cite{che2018}.}

I find that the fraction of people against integration still remains around 25\% (column 1). Moreover, the expected ranking of women hurt by integration increases, whereas that of men decreases (column 2). This finding suggests that the proposer's advantage becomes larger with correlated preferences. This advantage achieves its peak with $\rho=0.7$, but declines when the preferences become more correlated. This is due to the fact that the size of the set of stable matching collapses when the preferences are highly correlated \citep{holzman2014}.\footnote{The set of stable matchings even becomes a singleton for a variety of highly correlated preferences  \citep{eeckhout2000,clark2006,ortega2017}.} For the same reason, men's welfare losses become larger and comparable to those suffered by women as $\rho$ increases (column 3).

\newpage


\begin{thebibliography}{34}
\newcommand{\enquote}[1]{``#1''}
\expandafter\ifx\csname natexlab\endcsname\relax\def\natexlab#1{#1}\fi

\bibitem[\protect\citeauthoryear{Abdulkadiro{\u{g}}lu and
  S{\"o}nmez}{Abdulkadiro{\u{g}}lu and S{\"o}nmez}{2003}]{abdulkadirouglu2003}
\textsc{Abdulkadiro{\u{g}}lu, A. and T.~S{\"o}nmez} (2003): \enquote{School
  choice: A mechanism design approach,} \emph{American Economic Review}, 93,
  729--747.

\bibitem[\protect\citeauthoryear{Ashlagi, Kanoria, and Leshno}{Ashlagi
  et~al.}{2017}]{ashlagi2017}
\textsc{Ashlagi, I., Y.~Kanoria, and J.~Leshno} (2017): \enquote{Unbalanced
  random matching markets: The stark effect of competition,} \emph{Journal of
  Political Economy}, 125, 69--98.

\bibitem[\protect\citeauthoryear{Ashlagi and Roth}{Ashlagi and
  Roth}{2014}]{ashlagi2014}
\textsc{Ashlagi, I. and A.~Roth} (2014): \enquote{Free riding and participation
  in large scale, multi-hospital kidney exchange,} \emph{Theoretical
  Economics}, 9, 817--863.

\bibitem[\protect\citeauthoryear{Boudreau and Knoblauch}{Boudreau and
  Knoblauch}{2010}]{boudreau2010}
\textsc{Boudreau, J. and V.~Knoblauch} (2010): \enquote{Marriage matching and
  intercorrelation of preferences,} \emph{Journal of Public Economic Theory},
  12, 587--602.

\bibitem[\protect\citeauthoryear{Caldarelli and Capocci}{Caldarelli and
  Capocci}{2001}]{caldarelli2001}
\textsc{Caldarelli, G. and A.~Capocci} (2001): \enquote{Beauty and distance in
  the stable marriage problem,} \emph{Physica A: Statistical Mechanics and its
  Applications}, 300, 325--331.

\bibitem[\protect\citeauthoryear{Canny}{Canny}{2001}]{canny2001}
\textsc{Canny, J.} (2001): \enquote{Lecture notes in Combinatorics and Discrete
  Probability,} UC Berkeley,
  \url{https://people.eecs.berkeley.edu/~jfc/cs174/lecs/lec7/lec7.pdf}. Last
  visited on 2018/03/01.

\bibitem[\protect\citeauthoryear{Chambers and Hayashi}{Chambers and
  Hayashi}{2017}]{chambers2017}
\textsc{Chambers, C. and T.~Hayashi} (2017): \enquote{Can everyone benefit from
  economic integration?} \emph{unpublished}.

\bibitem[\protect\citeauthoryear{Che and Tercieux}{Che and
  Tercieux}{2018}]{che2018}
\textsc{Che, Y.-K. and O.~Tercieux} (2018): \enquote{{Efficiency and stability
  in large matching markets},} \emph{Journal of Political Economy},
  forthcoming.

\bibitem[\protect\citeauthoryear{Clark}{Clark}{2006}]{clark2006}
\textsc{Clark, S.} (2006): \enquote{{The uniqueness of stable matchings},}
  \emph{The B.E. Journal of Theoretical Economics}, 6, 1--28.

\bibitem[\protect\citeauthoryear{Clotfelter}{Clotfelter}{2004}]{clotfelter2004}
\textsc{Clotfelter, C.} (2004): \emph{After ``Brown'': The rise and retreat of
  school desegregation}, Princeton University Press.

\bibitem[\protect\citeauthoryear{Crawford}{Crawford}{1991}]{crawford1991}
\textsc{Crawford, V.} (1991): \enquote{Comparative statics in matching
  markets,} \emph{Journal of Economic Theory}, 54, 389--400.

\bibitem[\protect\citeauthoryear{Do\u{g}an and Yenmez}{Do\u{g}an and
  Yenmez}{2017}]{dogan2017}
\textsc{Do\u{g}an, B. and B.~Yenmez} (2017): \enquote{Unified enrollment in
  school choice: How to improve student assignment in Chicago,} \emph{SSRN
  preprint}.

\bibitem[\protect\citeauthoryear{Eeckhout}{Eeckhout}{2000}]{eeckhout2000}
\textsc{Eeckhout, J.} (2000): \enquote{On the uniqueness of stable marriage
  matchings,} \emph{Economics Letters}, 69, 1--8.

\bibitem[\protect\citeauthoryear{Ekmekci and Yenmez}{Ekmekci and
  Yenmez}{2018}]{ekmekci2018}
\textsc{Ekmekci, M. and B.~Yenmez} (2018): \enquote{Integrating schools for
  centralized admissions,} \emph{SSRN preprint}.

\bibitem[\protect\citeauthoryear{Gale and Shapley}{Gale and
  Shapley}{1962}]{gale1962}
\textsc{Gale, D. and L.~Shapley} (1962): \enquote{College admissions and the
  stability of marriage,} \emph{American Mathematical Monthly}, 69, 9--15.

\bibitem[\protect\citeauthoryear{Hafalir and Yenmez}{Hafalir and
  Yenmez}{2017}]{hafalir2017}
\textsc{Hafalir, I. and B.~Yenmez} (2017): \enquote{Integrating school
  districts: Diversity, balance, and welfare,} \emph{SSRN preprint}.

\bibitem[\protect\citeauthoryear{Hafalir, Yenmez, and Yildirim}{Hafalir
  et~al.}{2013}]{hafalir2013}
\textsc{Hafalir, I., B.~Yenmez, and M.~Yildirim} (2013): \enquote{Effective
  affirmative action in school choice,} \emph{Theoretical Economics}, 8,
  325--363.

\bibitem[\protect\citeauthoryear{Holzman and Samet}{Holzman and
  Samet}{2014}]{holzman2014}
\textsc{Holzman, R. and D.~Samet} (2014): \enquote{Matching of like rank and
  the size of the core in the marriage problem,} \emph{Games and Economic
  Behavior}, 88, 277--285.

\bibitem[\protect\citeauthoryear{Kelso and Crawford}{Kelso and
  Crawford}{1982}]{kelso1982}
\textsc{Kelso, A. and V.~Crawford} (1982): \enquote{Job matching, coalition
  formation, and gross substitutes,} \emph{Econometrica}, 50, 1483--1504.

\bibitem[\protect\citeauthoryear{Knuth}{Knuth}{1997}]{knuth1997}
\textsc{Knuth, D.} (1997): \emph{Stable marriage and its relation to other
  combinatorial problems}, no.~10 in CRM Proceedings \& Lecture Notes, AMS.

\bibitem[\protect\citeauthoryear{Lee}{Lee}{2017}]{lee2017}
\textsc{Lee, S.} (2017): \enquote{Incentive compatibility of large centralized
  matching markets,} \emph{Review of Economic Studies}, 84, 444.

\bibitem[\protect\citeauthoryear{Manjunath and Turhan}{Manjunath and
  Turhan}{2016}]{manjunath2016}
\textsc{Manjunath, V. and B.~Turhan} (2016): \enquote{Two school systems, one
  district: What to do when a unified admissions process is impossible,}
  \emph{Games and Economic Behavior}, 95, 25--40.

\bibitem[\protect\citeauthoryear{McVitie and Wilson}{McVitie and
  Wilson}{1971}]{mcvitie1971}
\textsc{McVitie, D. and L.~Wilson} (1971): \enquote{The stable marriage
  problem,} \emph{Communications of the ACM}, 14, 486--490.

\bibitem[\protect\citeauthoryear{Motwani and Raghavan}{Motwani and
  Raghavan}{1995}]{motwani1995}
\textsc{Motwani, R. and P.~Raghavan} (1995): \emph{Randomized algorithms}, New
  York, NY, USA: Cambridge University Press.

\bibitem[\protect\citeauthoryear{Niederle and Roth}{Niederle and
  Roth}{2003}]{niederle2003}
\textsc{Niederle, M. and A.~Roth} (2003): \enquote{Unraveling reduces mobility
  in a labor market: Gastroenterology with and without a centralized match,}
  \emph{Journal of Political Economy}, 111, 1342--1352.

\bibitem[\protect\citeauthoryear{Ortega and Hergovich}{Ortega and
  Hergovich}{2017}]{ortega2017}
\textsc{Ortega, J. and P.~Hergovich} (2017): \enquote{The strength of absent
  ties: Social integration via online dating,} \emph{arXiv preprint
  1709.10478}.

\bibitem[\protect\citeauthoryear{Pittel}{Pittel}{1989}]{pittel1989}
\textsc{Pittel, B.} (1989): \enquote{The average number of stable matchings,}
  \emph{SIAM Journal on Discrete Mathematics}, 2, 530--549.

\bibitem[\protect\citeauthoryear{Roth and Peranson}{Roth and
  Peranson}{1999}]{roth1999}
\textsc{Roth, A. and E.~Peranson} (1999): \enquote{The redesign of the matching
  market for American physicians: Some engineering aspects of economic design,}
  \emph{American Economic Review}, 89, 748--780.

\bibitem[\protect\citeauthoryear{Roth and Sotomayor}{Roth and
  Sotomayor}{1992}]{roth1992}
\textsc{Roth, A. and M.~Sotomayor} (1992): \emph{Two-sided matching: A study in
  game-theoretic modeling and analysis}, Econometric Society Monographs,
  Cambridge University Press.

\bibitem[\protect\citeauthoryear{Sprumont}{Sprumont}{1990}]{sprumont1990}
\textsc{Sprumont, Y.} (1990): \enquote{Population monotonic allocation schemes
  for cooperative games with transferable utility,} \emph{Games and Economic
  Behavior}, 2, 378--394.

\bibitem[\protect\citeauthoryear{Thomson}{Thomson}{1983}]{thomson1983}
\textsc{Thomson, W.} (1983): \enquote{The fair division of a fixed supply among
  a growing population,} \emph{Mathematics of Operations Research}, 8,
  319--326.

\bibitem[\protect\citeauthoryear{Toda}{Toda}{2006}]{toda2006}
\textsc{Toda, M.} (2006): \enquote{Monotonicity and consistency in matching
  markets,} \emph{International Journal of Game Theory}, 34, 13.

\bibitem[\protect\citeauthoryear{Toulis and Parkes}{Toulis and
  Parkes}{2015}]{toulis2015}
\textsc{Toulis, P. and D.~Parkes} (2015): \enquote{Design and analysis of
  multi-hospital kidney exchange mechanisms using random graphs,} \emph{Games
  and Economic Behavior}, 91, 360--382.

\bibitem[\protect\citeauthoryear{Wilson}{Wilson}{1972}]{wilson1972}
\textsc{Wilson, L.} (1972): \enquote{An analysis of the stable marriage
  assignment algorithm,} \emph{BIT Numerical Mathematics}, 12, 569--575.

\end{thebibliography}
\end{document}